\documentclass{article}
\usepackage{macros}
\usepackage{authblk}

\begin{document}

\title{Polyregular equivalence is undecidable in higher-order types}
\author[1]{Miko{\l}aj Boja\'nczyk}
\author[2]{Grzegorz Fabia\'nski}
\author[2]{Rafa{\l} Stefa\'nski}
\affil[1]{University of Warsaw, Poland}
\affil[2]{IDEAS Research Institute, Poland}

\maketitle 
\begin{abstract}
  It is open whether equivalence $f = g$ is decidable for string-to-string polyregular functions. We consider their higher-order extension based on the $\lambda$-calculus definition of polyregular functions from Bojańczyk (2018). In this setting, equivalence is undecidable by reduction from the tiling problem.
\end{abstract}

\section{Introduction}
\label{sec:introduction}
This paper is about the equivalence problem in transducer theory. In this problem, we are given two functions specified by some kind of finite state model, and we want to know if they are extensionally equal, i.e.~for every input, the two outputs are equal. A simple example of this kind is the equivalence problem for Mealy machines, which are deterministic automata with transitions labelled by letters of an output alphabet. A Mealy machine defines a length-preserving string-to-string function. In this case, the equivalence problem is the same question as equivalence of deterministic finite automata that recognise languages, and so it is decidable in polynomial time. A more interesting result is decidability of equivalence for string-to-string functions defined by deterministic two-way transducers~\cite[Theorem 1]{gurariEquivalenceProblemDeterministic1982}, or for (possibly copyful) streaming string transducers~\cite[Theorem 3]{filiotCopyfulStreamingString2017}. These results concern deterministic models; for nondeterministic transducers equivalence is undecidable already in the most basic model of rational relations, i.e.~nondeterministic automata with outputs on transitions, which define binary relations on strings~\cite[Corollary 4.1]{sakarovitch2009elements}. (This undecidability result is already implicit in~\cite[Theorem 18]{RabinScott59}.)

One major open problems about decidability of equivalence concerns the class of polyregular functions~\cite{polyregular-survey}. 
 The most concise definition of this class, see~\cite[Section 3]{polyregular-survey},  is that it consists of finite compositions of functions which are either: (a) defined by a deterministic two-way transducer, or (b) the \emph{marked squaring function}, which is explained in the following example 
    \begin{align*}
        \label{eq:marked-squaring}
    1234 \quad \mapsto \quad \underline 1234 1\underline 234 12 \underline 34 123 \underline 4 .
    \end{align*}
There are several other computational models which are equivalent to the one described above, see the survey~\cite{polyregular-survey} for at  five examples. One of these models, which is a functional programming language, is the topic of this paper and will be described in more detail shortly. The outstanding open problem about polyregular functions is decidability of equivalence, as stated below.

 \begin{problem}\label{prob:polyregular-equivalence}
    Is the following problem decidable?
  \begin{itemize}
    \item \textbf{Input.} Two polyregular functions  $f,g : A^* \to B^*$, with $A$ and $B$ finite.
    \item \textbf{Question.} Are they equivalent, i.e.~the outputs equal for all input strings?
 \end{itemize}
 \end{problem}

 In this paper, we show that the problem becomes undecidable, once it is generalized to allow for higher-order functions. 
The higher-order generalization  is based on  the  \emph{polyregular $\lambda$-calculus}, which is a functional programming language  based on the simply typed $\lambda$-calculus~\cite[Section 4.1]{polyregular-survey}. This language was designed to define the polyregular functions. This means that for  programs which have a  \emph{string-to-string type}, i.e.~a type of the form 
\begin{align*}
A^* \to B^* \qquad \text{for $A$ and $B$ finite,}
\end{align*}
the programming language defines exactly the polyregular functions~\cite[Theorem 4.1]{polyregular-survey}. However, the  language also allows for higher-order types, e.g.~types of the form
\begin{align*}
(A^* \to B^*) \to C^*,
\end{align*}
or even more complicated types. For the intended application of the language, which is to define the polyregular string-to-string functions, these higher-order types are used for subroutines in programs which ultimately have a string-to-string type. However, one could also be interested in  the equivalence problem when such higher-order types are used for instances of the equivalence problem. The contribution of this paper is Theorem~\ref{thm:undecidability}, which shows that the equivalence problem becomes undecidable.

\section{The programming language}
\label{sec:programming-language}
\newcommand{\leftterm}{\mathtt{Left}}
\newcommand{\rightterm}{\mathtt{Right}}
\newcommand{\fstterm}{\mathtt{fst}}
\newcommand{\sndterm}{\mathtt{snd}}
\newcommand{\mapterm}{\mathtt{map}}
\newcommand{\casesterm}{\mathtt{either}}
\newcommand{\splitterm}{\mathtt{split}}
\newcommand{\multterm}{\mathtt{mult}}
\newcommand{\headtailterm}{\mathtt{headtail}}
\newcommand{\concatterm}{\mathtt{concat}}
\newcommand{\headterm}{\mathtt{head}}
\newcommand{\tailterm}{\mathtt{tail}}
\newcommand{\blockterm}{\mathtt{block}}

We begin by  briefly describing the programming language in question, which is the polyregular $\lambda$-calculus. 
This programming language can be seen as a fragment of Haskell, and we will use Haskell syntax for our code. (Somewhat inconsistently, we do not use Haskell notation for types, since write $A^*$ instead of $[A]$, and likewise for products, which we denote by $A \times B$.)  Let us first illustrate the programming language by an example.

\begin{myexample}[Unmarked squaring]\label{ex:squaring} Consider some type $A$, e.g.~a finite alphabet. 
Here is a program which implements the  function 
\begin{align*}
\mathtt{square} : A^* & \to A^* \\
w & \mapsto w^{|w|}.
\end{align*}
This is similar to the marked squaring function mentioned in the introduction, except that it does not use underlines. Here is the program:
\begin{lstlisting}[language=Haskell]
  \w -> concat (map (\c -> w) w)
\end{lstlisting}
This program inputs a list \texttt{w}, replaces every character \texttt{c} in this list by the whole list itself, and concatenates the resulting nested list to get a non-nested list. 
\end{myexample}

As is often the case, a fully formal and rigorous definition of the programming language would take too much space, and so we only give an informal description here, referring the reader to~\cite[Section 4]{polyregular-survey} for more details. The language uses types built from the following  type constructors:
\begin{align*}
    \myunderbrace{1}{unit}
    \qquad 
    \myunderbrace{A \times B}{product} \qquad 
    \myunderbrace{A + B}{disjoint\\ union} \qquad 
    \myunderbrace{A^*}{lists} \qquad 
    \myunderbrace{A \to B}{functions}.
\end{align*}
The unit type is the atomic type, and it is meant to represent a set with one element, which is denoted by (). We use Church style typing, which means that each variable has an associated type. 

To construct terms we can use $\lambda$-abstraction and function application (as was illustrated in Example~\ref{ex:squaring}), as well as basic constructors for data, namely 
\begin{align*}
\myunderbrace{()}{unit} \quad 
\myunderbrace{(M,N)}{pair} \quad 
\myunderbrace{\leftterm M \quad \rightterm N}{disjopint union} \quad \myunderbrace{[M_1,\ldots,M_n]}{list}.
\end{align*}
Furthermore, the language  has certain built-in functions, which are mainly deconstructors for the data, as well as functions for processing lists. These are given below using their standard Haskell names, with their semantics being the same as in Haskell:
 \begin{align*}
     \fstterm & :  A \times B \to A \\
     \sndterm & : A \times B \to B \\
        \mathtt{uncons} & : A^* \to 1 + A \times A^* \\
        \casesterm & : (A \to C) \to (B \to C) \to (A + B) \to C \\
    \concatterm & : A^{**} \to A^* \\
        \mapterm & : (A \to B) \to A^* \to B^* 
    \end{align*}
There are two other list processing functions in the polyregular $\lambda$-calculus, which are  called $\splitterm$ and group multiplication. Since these functions are not needed for the undecidability proof, we do not explain them here; the reader can find their definitions in~\cite[Section 4]{polyregular-survey}.

Importantly, the programming language has no features for iteration or recursion. In particular, the programs are guaranteed to terminate, and therefore we can use a simple denotational semantics, in which each type $A$ is interpreted naively as a semantic domain $\sem A$. The semantic domain $\sem 1$ is  a singleton set, the interpretations of $A \times B$, $A + B$ and $A^*$ are the standard ones, and $\sem {A \to B}$ is the set of all total functions from $\sem A$ to $\sem B$. To each program $M$ of type $A$, we assign its semantics $\sem M \in \sem A$ in the natural way.
\section{Undecidability}
\label{sec:undecidability}
In this section, we present our contribution, which is undecidability of equivalence for the polyregular $\lambda$-calculus, when instances of higher-order type are allowed.

\begin{theorem}\label{thm:undecidability}
    The following problem is undecidable. 
    \begin{itemize}
        \item \textbf{Input.} Two programs $M$ and $N$ of the polyregular $\lambda$-calculus of the same type, which is not necessarily a string-to-string type.
        \item \textbf{Question.} Are the programs equivalent, i.e.~$\sem M = \sem N$?
    \end{itemize}
\end{theorem}
\begin{proof}
The proof is by a reduction from the tiling problem. The essential idea is very simple: a program can input a potential solution to the tiling problem and check if this solution is indeed valid. The potential solution is itself a function, and hence the need for higher-order types. We begin by formally describing the tiling problem. 

\paragraph*{The tiling problem.} We use a  variant of the tiling problem where  a solution that is a finite $n \times n$ square, with some prescribed tiles in two opposing corners. Let us explain this in more detail. An \emph{instance} of the tiling problem consists of: (a) a finite set  of tiles, with each tile having colours from some finite set  assigned to its four sides;  (b) a designated tile for the upper-left corner; and (c) a designated tile for the lower-left corner. Here is an example of such an instance:
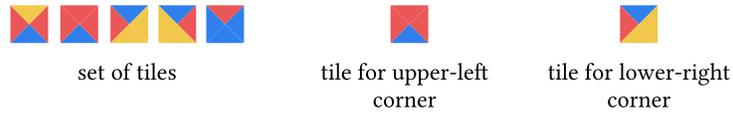
\mypic{1}
A \emph{solution} for such an instance is an $n \times n$ square labelled by the tiles, such that the colours on adjacent sides match (both horizontally and vertically), and the top-left and bottom-right corners have the designated tiles. Here is a picture:
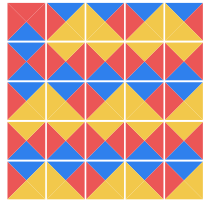
\mypic{2}
The tiling problem is to decide if an instance has a solution.  This is a  classical undecidable problem, with undecidability proved by using rows in the tiling to represent configurations of a Turing machine. 

\paragraph*{The reduction from tiling.} We show that for every instance of the tiling problem, one can compute a pair of polyregular programs, in a higher-order type, which are equivalent if and only if the tiling instance has no solution. This will imply undecidability of the tiling problem. 

Fix  an instance of the tiling problem. We  will use these auxiliary types: 
\begin{align*}
\myoverbrace{\Nat \eqdef 1^*}{natural numbers \\ are encoded in unary}\qquad \Bool \eqdef 1 + 1
\qquad T \eqdef \myoverbrace{ 1 + \cdots + 1}{number of tiles}
\end{align*}
To describe the non-existence of a solution, we will write programs of type 
\begin{align*}
 \myunderbrace{\Nat}{size\\  of   side in \\ solution} \times \ \myunderbrace{(\Nat \times \Nat \to T)}{the solution, extended \\ arbitrarily to the \\ infinite quarter-plane} \  \to  \  \Bool.
\end{align*}
We will write two programs of this type. The first program will always return ``false''. The second program will input a size $n : \Nat$ and a function $\Tt : \Nat \times \Nat \to T$, and it  will return ``true'' if and only if the function $\Tt$ describes a solution,  when restricted to the square that consists of rows and columns in $\set{0,\ldots,n}$.  These two programs will be equivalent if and only if the instance has no solution, thus proving undecidability of the equivalence problem. Here is the code of the second program:
\begin{lstlisting}[language=Haskell]
  \(n, T) ->
    ((T (0,0)) = upperLeftCornerTile) &&
    ((T (n,n)) = lowerRightCornerTile) &&
    (all matchHorizontally (horizontalWindows n T)) &&
    (all matchVertically (verticalWindows n T))
\end{lstlisting}
The essential idea is that lines 2 and 3 check the corner tiles, while lines 4 and 5 check that adjacent tiles are matching on their colours. This is implemented by looking at two-tile windows,  either horizontal or vertical, as in the following picture: 
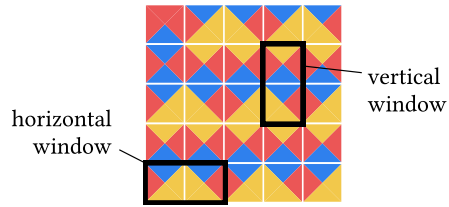
\mypic{3}
Although we hope that the program is self-explanatory, we explain its subroutines in more detail below, with an emphasis on showing that they can be expressed in  the polyregular $\lambda$-calculus, without using recursion of unrestricted Haskell. 
\begin{itemize}
        \item  In lines 1 and 2 of the program we have  constants 
    \begin{align*}
    \mathtt{upperLeftCornerTile}, \mathtt{lowerRightCornerTile} : T
    \end{align*}
    represent the designated corner tiles. The program uses also Boolean conjunction \&\& and equality on tiles. These can be implemented, since  the polyregular $\lambda$-calculus can implement any function $A \to B$ where both $A$ and $B$ have finite types. In the same way, we can implement  the functions 
        \begin{align*}
        \mathtt{matchHorizontally}, \mathtt{matchVertically} : T \times T \to \Bool
        \end{align*}
    which check if two windows are matching horizontally or vertically.   
    \item In lines 4 and 5 we use the function 
    \begin{align*}
    \mathtt{all} : (T \times T \to \Bool) \to (T \times T)^* \to \Bool,
    \end{align*}
    which tells us if all windows in a given list satisfy some predicate. This function can be implemented in the polyregular $\lambda$-calculus, by filtering out all the true elements
    and checking if the remaining list is empty~\cite[p.~89]{bojanczykPolyregularFunctions2018}.
    
    \item The last component we need to explain is the functions 
    \begin{align*}
    \mathtt{horizontalWindows}, \mathtt{verticalWindows} : \Nat \to (\Nat \times \Nat \to T) \to  (T \times T)^*,
    \end{align*}
    which return the two-tile windows in the horizontal or vertical direction, respectively.
    These are implemented in the obvious way using list comprehension:
    \begin{lstlisting}[language=Haskell]
    horizontalWindows n T =
      concat (map (\i -> 
        concat (map (\j -> [(T(i,j),T(i+1,j))]) (range n))
          ) (range n))
    \end{lstlisting}
    In the above code, the \texttt{range} subroutine maps $n$ to $[0,\ldots,n-1]$; such a function can easily be implemented in the polyregular $\lambda$-calculus~\cite[Example 13]{polyregular-survey}. 
\end{itemize}
\end{proof}

\section{Conclusion}
\label{sec:conclusion}
We have shown that equivalence for the polyregular $\lambda$-calculus is undecidable in higher-order types. The order of types in the undecidability is not very high, namely 2: these are functions that input functions. For order 1, i.e.~types of the form 
\begin{align*}
A_1 \to A_2 \to \cdots \to A_n, \qquad \text{where each $A_i$ is arrow-free,}
\end{align*}
the equivalence problem is the same as for string-to-string functions. Indeed, by using Currying, such functions can be seen as having type $A \to B$, where both $A$ and $B$ are arrow-free. Next, elements of an arrow-free type can be represented as strings over a finite alphabet, thus reducing the problem to the string-to-string case, see e.g.~\cite[Lemma 4.4]{bojanczykRegularFirstOrderList2018}.  We do hope that in order 1, i.e.~in the case of string-to-string functions, the equivalence problem is decidable, but this remains an open problem.

\bibliographystyle{plain}
\bibliography{bib}

\end{document}